\documentclass[11pt]{article}

\usepackage[a4paper,margin=1in]{geometry}
\usepackage[comma,authoryear,round]{natbib}
\usepackage{amsmath,amssymb,amsthm}
\usepackage{booktabs}
\usepackage{enumitem}
\usepackage{xcolor}
\usepackage{microtype}
\usepackage[unicode=true,bookmarks=false,breaklinks=true,
  colorlinks=true,citecolor=blue,linkcolor=blue,urlcolor=blue]{hyperref}

\newtheorem{theorem}{Theorem}

\theoremstyle{definition}
\newtheorem{definition}{Definition}

\theoremstyle{remark}
\newtheorem{remark}[theorem]{Remark}

\newcommand{\N}{N}
\newcommand{\M}{M}

\title{Pairwise Maximin Share Allocations Need Not Exist}
\author{Haris Aziz\\
UNSW Sydney, Australia\\
\href{mailto:haris.aziz@unsw.edu.au}{\texttt{haris.aziz@unsw.edu.au}}}
\date{}

\begin{document}
\maketitle

\begin{abstract}
We study the fair allocation of indivisible goods among agents with strictly
positive additive valuations. Pairwise maximin share fairness (PMMS) asks
that, for every ordered pair of agents, the first agent value her own bundle
at least as highly as the best worst-case share she could secure by
repartitioning the two agents' combined bundles into two parts.  Whether a PMMS allocation always exists for positive additive
valuations has remained open since the notion was introduced.
We resolve this question in the negative. We construct an instance with four
agents and strictly positive additive valuations that admits no complete PMMS
allocation. 
\end{abstract}

\section{Introduction}

Cut-and-choose is the canonical protocol for dividing a resource between two
agents. One agent partitions the resource into two parts and the other agent
chooses first. The cutter can therefore guarantee the value of her best
balanced two-way partition. For divisible resources this idea underlies
envy-free cake cutting. With indivisible goods, exact envy-freeness can fail
for the elementary reason that a single valuable good cannot be split.

\textit{Pairwise maximin share (PMMS)} fairness imports the robust part of cut-and-choose into
a multi-agent allocation. Fix two agents $i$ and $j$ and pool the bundles they
receive. Agent $i$ asks what value she could guarantee if she repartitioned
that pool into two bundles and received the less valuable one. An allocation
is PMMS if $i$'s actual bundle meets this benchmark for every other agent $j$~\citep{caragiannis2019unreasonable}.

This definition is appealing for several reasons. First, it has the direct
procedural interpretation of pairwise cut-and-choose. Second, it treats every
pair symmetrically while evaluating fairness from each agent's own cardinal
perspective. Third, for strictly positive additive valuations, PMMS implies
EFX and hence also EF1.
PMMS is thus a
particularly demanding local fairness benchmark, yet one that is meaningful
even when exact envy-freeness is impossible. Considering that the existence of an EFX allocation remains an open problem for positive additive valuations, it also motivates to settle the guaranteed existence of 
allocations satisfying the stronger PMMS concept. 

The strength of PMMS has made its existence difficult to understand. Exact
PMMS is immediate for two agents and for identical valuations, but the
general additive problem remained unresolved. Work therefore concentrated on
multiplicative approximations, culminating in a $0.781$ existential guarantee
for additive valuations \citep{kurokawa2017fair}. Exact existence was
subsequently established in several structured domains, while a recent
counterexample showed failure for monotone valuations when two of the three
valuations are nonadditive \citep{byrka2026probing}. None of these results
settled the additive case.

\newpage

\paragraph{Our contribution.}
We prove that a complete PMMS allocation need not exist even when every agent
has a strictly positive additive valuation.

\begin{theorem}[Main result]\label{thm:main}
There exists an instance with four agents and strictly positive additive
valuations for which no complete PMMS allocation exists.
\end{theorem}

Theorem~\ref{thm:main} resolves the exact-existence question for additive
valuations in the negative. Strict positivity is worth emphasizing: the
obstruction is not caused by dummy goods or zero-value degeneracies. Since a
PMMS allocation always exists for two additive agents and the three-agent
additive case remains open, four agents are the smallest number for which
additive nonexistence is currently known.

The result also clarifies the relationship with EFX. Under strictly positive
additive valuations, PMMS implies EFX, but the converse fails. Our theorem
therefore rules out PMMS as a route to proving universal additive EFX
existence; it does not provide an EFX counterexample. 


\section{Related Work}\label{sec:related}

\paragraph{Origins and neighboring notions.}
\citet{caragiannis2019unreasonable} introduced PMMS together with EFX in their study of the
fairness properties of maximum Nash welfare allocations. 
For a
broader survey of fairness notions for indivisible goods, see
\citet{amanatidis2023survey}.

For positive additive valuations, PMMS implies EFX. Quantitatively,
\citet{amanatidis2018comparing} showed that every EFX allocation is
$2/3$-PMMS and every EF1 allocation is $1/2$-PMMS; both implications are tight
in the appropriate sense. As EF1 allocations can be computed efficiently
\citep{lipton2004approximately}, the latter implication already gives a
simple polynomial-time baseline for approximate PMMS. In the other direction,
approximate PMMS alone need not give any constant EFX guarantee
\citep{amanatidis2018comparing}. Very recent work determines that exact PMMS
guarantees a tight $10/17$ fraction of the global MMS benchmark
\citep{qin2026exactmms}.

In the chore (additive negative valuations) setting, \citet{HeTa26a} prove that EFX and hence PMMS allocations need not exist for any ($n\geq 4$). 

\paragraph{Approximate PMMS for additive valuations.}
The best general existential guarantee is $0.781$
\citep{kurokawa2017fair}, improving the
$(\varphi-1)\approx0.618$ guarantee of Nash-welfare maximization
\citep{caragiannis2019unreasonable}. Simultaneous guarantees include an
efficiently computable $0.618$-EFX and $0.717$-PMMS allocation
\citep{amanatidis2020multiple}, and an existential allocation that is
$0.618$-EFX, $0.618$-PMMS, and achieves $0.618$ of the optimal Nash welfare
\citep{feldman2024tradeoffs}. Instances with range parameter $\gamma$ admit
efficiently computable $(5\gamma/6)$-PMMS allocations; in particular,
restricted additive instances admit efficiently computable $5/6$-PMMS
allocations \citep{barman2024parameterized}. Common-ranking instances admit
efficiently computable $4/5$-PMMS allocations, originally claimed by
\citet{dai2024existence} and subsequently proved using a nondegeneracy
reduction by \citet{feldman2026epmms}.

\paragraph{Exact existence and relaxations.}
Before the present work, no additive counterexample to exact PMMS was known.
\citet{byrka2026probing} constructed a three-agent instance with no PMMS
allocation, but only one valuation in that instance is additive; the other
two are monotone and nonadditive. The same paper proves polynomial-time
existence for factored personalized bivalued valuations, for binary-valued
MMS-feasible valuations, and for pair-demand valuations. In another sparse
domain, exact PMMS exists when each good is valued by at most two agents
\citep{christodoulou2026multigraph}.


\section{Model and Definitions}\label{sec:model}

Let $\N=\{1,\ldots,n\}$ be a finite set of agents and let
$\M=\{g_1,\ldots,g_m\}$ be a finite set of indivisible goods. Each agent
$i\in\N$ has a valuation function $v_i:2^\M\to\mathbb{R}_{\geq0}$, normalized
so that $v_i(\varnothing)=0$.

The valuation $v_i$ is \emph{additive} if
$
  v_i(S)=\sum_{g\in S}v_i(g)
  \qquad\text{for every }S\subseteq\M,
$
where $v_i(g)$ abbreviates $v_i(\{g\})$. It is \emph{strictly positive} if
$v_i(g)>0$ for every $g\in\M$. Unless stated otherwise, all valuations in our
counterexample are additive and strictly positive.

An \emph{allocation} is a tuple $A=(A_1,\ldots,A_n)$ of pairwise disjoint
bundles. It is \emph{complete} if $\bigcup_{i\in\N}A_i=\M$. All allocations
in the main result are complete.

\begin{definition}[Pairwise maximin share fairness]\label{def:pmms}
	For $\alpha\in[0,1]$, an allocation $A=(A_1,\ldots,A_n)$ is
	\emph{$\alpha$-pairwise maximin share fair} ($\alpha$-PMMS) if, for every
	pair of distinct agents $i,j\in\N$,
	
	$$
	v_i(A_i)\geq
	\alpha\max_{B\subseteq A_i\cup A_j}
	\min\bigl\{v_i(B),v_i((A_i\cup A_j)\setminus B)\bigr\}.
	$$
	
	It is \emph{pairwise maximin share fair} (PMMS) when the condition holds
	with $\alpha=1$. Otherwise, if the inequality fails for $(i,j)$, agent $i$
	\emph{PMMS-envies} agent $j$.
\end{definition}

%
%

For comparison, an allocation is EFX if
$
  v_i(A_i)\geq v_i(A_j\setminus\{g\})
\text{ for every }i,j\in\N\text{ and every }g\in A_j.
$
Because our valuations are strictly positive, there is no distinction here
between EFX and the convention that quantifies only over goods of positive
value to the potentially envious agent.

\section{The Counterexample}\label{sec:counterexample}

We now give the instance establishing Theorem~\ref{thm:main}.  The proof is
self-contained and uses only the values in Table~\ref{tab:counterexample},
additivity, and the definition of PMMS.

There are four agents and $52$ goods.  Partition the goods into four classes
$M_A,M_B,M_C,M_D$ of respective sizes $3,5,5,39$.  Goods in the same class
have the same value.  Table~\ref{tab:counterexample} gives the value of one
good from each class.

\begin{table}[ht]
\centering
\caption{The strictly positive additive counterexample.}
\label{tab:counterexample}
\begin{tabular}{@{}lrrrr@{}}
\toprule
& $M_A$ ($3$ goods) & $M_B$ ($5$ goods)
& $M_C$ ($5$ goods) & $M_D$ ($39$ goods) \\
\midrule
Agents $1,2$ & $81$ & $100$ & $94$ & $101$ \\
Agents $3,4$ & $81$ & $94$  & $100$ & $101$ \\
\bottomrule
\end{tabular}
\end{table}

It is useful to express values relative to the common baseline $101$.  For
each agent $i$ and good $g$, define its \emph{deficit} by
\[
    c_i(g):=101-v_i(g),
\]
and extend $c_i$ additively to bundles.  Thus the deficits are
\begin{equation}
\begin{array}{c|rrrr}
& M_A&M_B&M_C&M_D\\ \hline
i\in\{1,2\}&20&1&7&0\\
i\in\{3,4\}&20&7&1&0
\end{array}
\label{eq:deficit-table}
\end{equation}
The total deficit of all goods is the same for every agent:
\begin{equation}
  c_i(\M)=3\cdot20+5\cdot1+5\cdot7+39\cdot0=100.
\label{eq:total-deficit}
\end{equation}

Since each singleton value satisfies $v_i(g)=101-c_i(g)$ and both $v_i$ and
$c_i$ are additive, summing over the goods in a bundle gives, for every
$S\subseteq\M$,
\begin{equation*}
	v_i(S)=\sum_{g\in S}\bigl(101-c_i(g)\bigr)=101|S|-c_i(S).
\end{equation*}
The deficit of a bundle is bounded on both sides. Every singleton deficit is
nonnegative, so $c_i(S)\geq 0$; nonnegativity also gives monotonicity,
$c_i(S)\leq c_i(\M)$. 
Hence $0\leq c_i(S)\leq c_i(M)=100$ for every bundle $S$, and
substituting into the identity above yields
\begin{equation}
	v_i(S)=101|S|-c_i(S),
	\qquad
	101|S|-100\leq v_i(S)\leq101|S|.
	\label{eq:value-deficit}
\end{equation}
The upper bound is attained exactly when $S$ consists of padding goods only,
and the lower bound exactly when $S$ contains all thirteen goods of
$A\cup B\cup C$. The essential feature of \eqref{eq:value-deficit} is that
the total deficit, $100$, is \emph{strictly smaller} than the base value
$101$ of a single good: a bundle is worth $101$ per good up to an error
smaller than one good's worth. Consequently a bundle with more goods is
always strictly more valuable---a set of $p+1$ goods is worth at least
$101p+1$, whereas any set of $p$ goods is worth at most $101p$---and this is
what lets cardinality dominate composition throughout the proof.

In particular, among bundles of the same size, agent $i$ prefers precisely
the bundle with the smaller $i$-deficit.

\begin{proof}[Proof of Theorem~\ref{thm:main}]
Suppose, for contradiction, that $X=(X_1,X_2,X_3,X_4)$ is a complete PMMS
allocation.

\textbf{Step 1: every bundle has size 13.}
If $|X_i|=p$ and $|X_j|\geq p+2$, then $X_i\cup X_j$ contains at least
$2p+2$ goods and admits a bipartition into parts of at least $p+1$ goods
each. By \eqref{eq:value-deficit}, each part has $i$-value at least
$101(p+1)-100=101p+1>101p\geq v_i(X_i)$. Hence
\[
\max_{B\subseteq X_i\cup X_j}
\min\{v_i(B),v_i((X_i\cup X_j)\setminus B)\}>v_i(X_i),
\]
contradicting PMMS. Bundle sizes therefore differ by at most one, and since
they sum to $52$, every bundle has exactly $13$ goods. Put
$R:=M_A\cup M_B\cup M_C$ and $Y_i:=X_i\cap R$.

\textbf{Step 2: a necessary deficit inequality.}
For every ordered pair $(i,j)$ and every bipartition $(U,V)$ of
$Y_i\cup Y_j$,
\begin{equation}
  \max\{c_i(U),c_i(V)\}\geq c_i(Y_i).
  \label{eq:transfer}
\end{equation}
Otherwise both deficits are strictly below $c_i(Y_i)$. Since
$Y_i\cup Y_j\subseteq R$ and $|R|=13$, both $|U|$ and $|V|$ are at most
$13$. The union
$X_i\cup X_j$ contains exactly
\[
26-|Y_i\cup Y_j|=(13-|U|)+(13-|V|)
\]
padding goods, so $(U,V)$ can be changed to a bipartition $(P,Q)$ of
$X_i\cup X_j$ with $|P|=|Q|=13$. Padding goods have deficit zero, so
\begin{align*}
v_i(P)&=13\cdot101-c_i(U)>13\cdot101-c_i(Y_i)=v_i(X_i),\\
v_i(Q)&=13\cdot101-c_i(V)>13\cdot101-c_i(Y_i)=v_i(X_i),
\end{align*}
contradicting PMMS.

Call the goods in $M_B\cup M_C$ \emph{ordinary}; their total deficit is
$5\cdot1+5\cdot7=40$ for every agent, while each $M_A$-good has deficit
$20$.

\textbf{Step 3: no agent receives two $M_A$-goods.}
Suppose agent $i$ receives $a\geq2$ of the $M_A$-goods, and let $s$ be the
$i$-deficit of the ordinary goods in $Y_i$, so
$c_i(Y_i)=20a+s\geq40$. At least two agents receive no $M_A$-good. Since
the ordinary goods have total deficit $40$, one of them, say $j$, satisfies
$t:=c_i(Y_j)\leq20$.

Choose $h\in Y_i\cap M_A$ and partition $Y_i\cup Y_j$ into
\[
U=Y_i\setminus\{h\},
\qquad
V=Y_j\cup\{h\}.
\]
Then
\[
c_i(U)=c_i(Y_i)-20<c_i(Y_i)
\]
and
\[
c_i(V)=20+t\leq40\leq c_i(Y_i).
\]
Thus both deficits are strictly below $c_i(Y_i)$ unless all three equalities
$a=2$, $s=0$, and $t=20$ hold. Outside this exceptional case,
\eqref{eq:transfer} is contradicted.

In the exceptional case, $Y_i$ consists of two $M_A$-goods, say $h_1,h_2$,
and $Y_j$ is a nonempty set of ordinary goods of total $i$-deficit $20$.
Pick $g\in Y_j$ and write $z:=c_i(g)\in\{1,7\}$. Partition
$Y_i\cup Y_j$ instead into
\[
\{h_1,g\}
\qquad\text{and}\qquad
\{h_2\}\cup(Y_j\setminus\{g\}).
\]
Their deficits are
\[
20+z\leq27<40=c_i(Y_i)
\]
and
\[
20+(20-z)=40-z<40=c_i(Y_i),
\]
again contradicting \eqref{eq:transfer}. Hence every agent receives at most
one $M_A$-good.

The three $M_A$-goods therefore lie with three distinct agents, and exactly
one agent receives none. By exchanging agents within the two identical
types, and if necessary exchanging the two types together with the names
$M_B,M_C$, we may assume agent $1$ receives no $M_A$-good. Set
\[
S:=Y_1,
\qquad
x:=|S\cap M_B|,
\qquad
y:=|S\cap M_C|.
\]
Thus
\[
c_1(S)=c_2(S)=x+7y,
\qquad
c_3(S)=c_4(S)=7x+y.
\]

\textbf{Step 4: $x\geq3$ and $y\geq3$.}
If $x+y\geq8$, then $x,y\leq5$ immediately gives $x,y\geq3$. Suppose
therefore that $x+y\leq7$. At least three ordinary goods then lie outside
$S$, all held by agents $2,3,4$.

First assume some $j\in\{2,3,4\}$ holds nothing beyond her $M_A$-good. The
other two $M_A$-holders then share the ordinary goods outside $S$, so one of
them, say $k$, holds at least two. Split the ordinary goods of $Y_k$ into two
nonempty parts of $k$-deficits $t_1,t_2\geq1$, and pair one of the two
$M_A$-goods of $Y_k\cup Y_j$ with each part. The resulting bipartition has
$k$-deficits
\[
20+t_1<20+t_1+t_2=c_k(Y_k)
\]
and
\[
20+t_2<20+t_1+t_2=c_k(Y_k),
\]
contradicting \eqref{eq:transfer}.

Hence every $j\in\{2,3,4\}$ holds an ordinary good, so
$\sigma_j:=c_j(Y_j)-20\geq1$. Partition $Y_j\cup S$ into the $M_A$-good of
$j$ and everything else. The two $j$-deficits are
\[
20<20+\sigma_j=c_j(Y_j)
\]
and
\[
\sigma_j+c_j(S).
\]
By \eqref{eq:transfer}, the second part must have deficit at least
$c_j(Y_j)=20+\sigma_j$. Therefore
\[
\sigma_j+c_j(S)\geq20+\sigma_j,
\qquad\text{so}\qquad c_j(S)\geq20.
\]
For $j=2$ and $j=3$, respectively, this gives
\[
x+7y\geq20,
\qquad
7x+y\geq20.
\]
Since $x,y\leq5$,
\[
x\leq2\ \Longrightarrow\ 7x+y\leq14+5=19,
\]
so $x\geq3$; symmetrically,
\[
y\leq2\ \Longrightarrow\ x+7y\leq5+14=19,
\]
so $y\geq3$.

\textbf{Step 5: the final violating partition.}
Because $y\geq3$, at most two $M_C$-goods lie outside $S$, so some
$j\in\{2,3,4\}$ holds none. Because $x\geq3$, at most two $M_B$-goods lie
outside $S$, and this $j$ holds at most two of them. Hence
\[
c_1(Y_j)\leq20+2\cdot1=22.
\]
Also
\[
c_1(S)=x+7y\geq3+7\cdot3=24.
\]
Since $x,y\geq3$, choose $T\subseteq S$ consisting of three $M_C$-goods and
two $M_B$-goods. Then
\[
c_1(T)=3\cdot7+2\cdot1=23<24\leq c_1(S).
\]
Partition $S\cup Y_j$ into $T$ and its complement. The complementary part
has deficit
\begin{align*}
c_1((S\cup Y_j)\setminus T)
 &=c_1(S)+c_1(Y_j)-c_1(T)\\
 &\leq c_1(S)+22-23\\
 &=c_1(S)-1\\
 &<c_1(S).
\end{align*}
Thus both parts have deficit strictly below $c_1(Y_1)=c_1(S)$, violating
\eqref{eq:transfer} for the ordered pair $(1,j)$ and completing the
contradiction.
\end{proof}

\begin{remark}[Origin of the construction]\label{rem:origin}
The $13$-item deficit pattern in~(1) is a modification of the additive-chore
instance used by \citet{HeTa26a}.  Their
instance has three chores of cost $20$, five chores with costs $1$ and $7$
for the two agent types, and five chores with the two costs reversed.  We
turn an item of cost $c_i(g)$ into a good of value $101-c_i(g)$ and add $39$
common goods of value $101$.  The added goods force every PMMS allocation to
have equal bundle sizes, after which value comparisons become the reverse of
the corresponding deficit comparisons.  
\end{remark}

%

\paragraph{Acknowledgement.}
The author acknowledges the assistance of OpenAI's GPT-5.6, which identified the counterexample.  All arguments and calculations  were independently verified by the author.

\bibliographystyle{plainnat}

\begin{thebibliography}{14}
	\providecommand{\natexlab}[1]{#1}
	\providecommand{\url}[1]{\texttt{#1}}
	\expandafter\ifx\csname urlstyle\endcsname\relax
	\providecommand{\doi}[1]{doi: #1}\else
	\providecommand{\doi}{doi: \begingroup \urlstyle{rm}\Url}\fi
	
	\bibitem[Amanatidis et~al.(2018)Amanatidis, Birmpas, and
	Markakis]{amanatidis2018comparing}
	Georgios Amanatidis, Georgios Birmpas, and Vangelis Markakis.
	\newblock Comparing approximate relaxations of envy-freeness.
	\newblock In \emph{Proceedings of the Twenty-Seventh International Joint
		Conference on Artificial Intelligence (IJCAI)}, pages 42--48, 2018.
	\newblock \doi{10.24963/ijcai.2018/6}.
	
	\bibitem[Amanatidis et~al.(2020)Amanatidis, Markakis, and
	Ntokos]{amanatidis2020multiple}
	Georgios Amanatidis, Evangelos Markakis, and Apostolos Ntokos.
	\newblock Multiple birds with one stone: Beating $1/2$ for {EFX} and {GMMS} via
	envy cycle elimination.
	\newblock \emph{Theoretical Computer Science}, 841:\penalty0 94--109, 2020.
	\newblock \doi{10.1016/j.tcs.2020.07.006}.
	
	\bibitem[Amanatidis et~al.(2023)Amanatidis, Aziz, Birmpas, Filos-Ratsikas, Li,
	Moulin, Voudouris, and Wu]{amanatidis2023survey}
	Georgios Amanatidis, Haris Aziz, Georgios Birmpas, Aris Filos-Ratsikas, Bo~Li,
	Herv{\'e} Moulin, Alexandros~A. Voudouris, and Xiaowei Wu.
	\newblock Fair division of indivisible goods: Recent progress and open
	questions.
	\newblock \emph{Artificial Intelligence}, 322:\penalty0 103965, 2023.
	\newblock \doi{10.1016/j.artint.2023.103965}.
	
	\bibitem[Barman et~al.(2024)Barman, Kar, and Pathak]{barman2024parameterized}
	Siddharth Barman, Debajyoti Kar, and Shraddha Pathak.
	\newblock Parameterized guarantees for almost envy-free allocations.
	\newblock \emph{CoRR}, abs/2312.13791, 2024.
	\newblock URL \url{https://arxiv.org/abs/2312.13791}.
	
	\bibitem[Byrka et~al.(2026)Byrka, Malinka, and Ponitka]{byrka2026probing}
	Jaros{\l}aw Byrka, Franciszek Malinka, and Tomasz Ponitka.
	\newblock Probing {EFX} via {PMMS}: (non-)existence results in discrete fair
	division.
	\newblock In \emph{Proceedings of the Fortieth AAAI Conference on Artificial
		Intelligence (AAAI)}, pages 16735--16742, 2026.
	\newblock URL \url{https://arxiv.org/abs/2507.14957}.
	
	\bibitem[Caragiannis et~al.(2019)Caragiannis, Kurokawa, Moulin, Procaccia,
	Shah, and Wang]{caragiannis2019unreasonable}
	Ioannis Caragiannis, David Kurokawa, Herv{\'e} Moulin, Ariel~D. Procaccia,
	Nisarg Shah, and Junxing Wang.
	\newblock The unreasonable fairness of maximum {Nash} welfare.
	\newblock \emph{ACM Transactions on Economics and Computation}, 7\penalty0
	(3):\penalty0 12:1--12:32, 2019.
	\newblock \doi{10.1145/3355902}.
	
	\bibitem[Christodoulou and Mastrakoulis(2026)]{christodoulou2026multigraph}
	George Christodoulou and Symeon Mastrakoulis.
	\newblock Exact and approximate maximin share allocations in multi-graphs.
	\newblock In \emph{Proceedings of the Fortieth AAAI Conference on Artificial
		Intelligence (AAAI)}, pages 16761--16769, 2026.
	\newblock URL \url{https://arxiv.org/abs/2506.20317}.
	
	\bibitem[Dai et~al.(2024)Dai, Guo, Miao, Gao, Xu, and Zhang]{dai2024existence}
	Sijia Dai, Xinru Guo, Huahua Miao, Guichen Gao, Yicheng Xu, and Yong Zhang.
	\newblock The existence and efficiency of {PMMS} allocations.
	\newblock \emph{Theoretical Computer Science}, 989:\penalty0 114388, 2024.
	\newblock \doi{10.1016/j.tcs.2024.114388}.
	
	\bibitem[Feldman et~al.(2024)Feldman, Mauras, and
	Ponitka]{feldman2024tradeoffs}
	Michal Feldman, Simon Mauras, and Tomasz Ponitka.
	\newblock On optimal tradeoffs between {EFX} and {Nash} welfare.
	\newblock In \emph{Proceedings of the Thirty-Eighth AAAI Conference on
		Artificial Intelligence (AAAI)}, pages 9688--9695, 2024.
	\newblock \doi{10.1609/aaai.v38i9.28825}.
	
	\bibitem[Feldman et~al.(2026)Feldman, Fiat, Nissan, and
	Ponitka]{feldman2026epmms}
	Michal Feldman, Amos Fiat, Yael Nissan, and Tomasz Ponitka.
	\newblock Epistemic pairwise maximin share.
	\newblock \emph{CoRR}, abs/2606.18921, 2026.
	\newblock URL \url{https://arxiv.org/abs/2606.18921}.
	
	\bibitem[He and Tao(2026)]{HeTa26a}
	W.~He and B.~Tao.
	\newblock {EFX} for additive chores: Nonexistence, {P}areto incompatibility,
	and bi-valued existence.
	\newblock \emph{arXiv preprint arXiv:2606.08872}, 2026.
	
	\bibitem[Kurokawa(2017)]{kurokawa2017fair}
	David Kurokawa.
	\newblock \emph{Fair Division in Game Theoretic Settings}.
	\newblock Ph.d. thesis, Carnegie Mellon University, 2017.
	\newblock URL
	\url{https://reports-archive.adm.cs.cmu.edu/anon/2017/CMU-CS-17-122.pdf}.
	
	\bibitem[Lipton et~al.(2004)Lipton, Markakis, Mossel, and
	Saberi]{lipton2004approximately}
	Richard~J. Lipton, Evangelos Markakis, Elchanan Mossel, and Amin Saberi.
	\newblock On approximately fair allocations of indivisible goods.
	\newblock In \emph{Proceedings of the 5th ACM Conference on Electronic Commerce
		(EC)}, pages 125--131, 2004.
	\newblock \doi{10.1145/988772.988792}.
	
	\bibitem[Qin(2026)]{qin2026exactmms}
	Qinghua Qin.
	\newblock The exact {MMS} guarantees of {EFX} and {PMMS}.
	\newblock \emph{CoRR}, abs/2608.30267, 2026.
	\newblock URL \url{https://arxiv.org/abs/2608.30267}.
	
\end{thebibliography}

\end{document}